\documentclass[10pt,twoside,a4paper,reqno]{amsart}
\usepackage{amsfonts,amsthm,latexsym,amsmath,amssymb,amscd,epsf}
\usepackage{graphicx}
\usepackage{float}

\theoremstyle{plain}
\newtheorem{theo}           {Theorem}
\newtheorem{pro}            {Proposition}
\newtheorem{coro}           {Corollary}
\newtheorem{lemm}           {Lemma}
\newtheorem{conj}           {Conjecture}
\newtheorem{defi}           {Definition}

\theoremstyle{definition}
\newtheorem{pr}              {Problem}
\newtheorem*{ack}            {Acknowledgements}

\newtheorem{nota}            {Notation}

\theoremstyle{remark}
\newtheorem{rem}             {Remark}

\newenvironment{theorem}{\begin{theo}}{\end{theo}}
\newenvironment{proposition}{\begin{pro}}{\end{pro}}
\newenvironment{corollary}{\begin{coro}}{\end{coro}}
\newenvironment{lemma}{\begin{lemm}}{\end{lemm}}
\newenvironment{remark}{\begin{rem}}{\end{rem}}
\newenvironment{definition}{\begin{defi}}{\end{defi}}

\newenvironment{problem}{\begin{pr}}{\end{pr}}
\newenvironment{conjecture}{\begin{conj}}{\end{conj}}

\newcommand {\dq}{\mathfrak d}
\newcommand {\C} {\mathcal C}
\newcommand {\CT} {\operatorname {CT}}
\DeclareMathOperator{\Conv}{Conv}

\begin{document}

\title[Van Vleck spectra of high-order Heun operators]
{Van Vleck spectra of high-order Heun operators:\ finite-band universality and exterior asymptotics}

\author{Boris Shapiro}
\address[Boris Shapiro]{
	Department of Mathematics, 
	Stockholm University,
	Kr\"{a}ftriket 5,
	SE - 106 91 Stockholm, Sweden}
\email{shapiro@math.su.se}

\keywords{generalized Lam\'e equation, Van Vleck and Heine-Stieltjes polynomials, 
asymptotic root distribution}
\subjclass[2020] {Primary 31A35, 34M35; Secondary 34E20, 30C15, 47B36}

\begin{abstract}
We study high-order analogues of the classical Heun operator of Fuchs index one,
\[
        \dq=\sum_{i=1}^k Q_i(z)\frac{d^i}{dz^i},
        \qquad \deg Q_i\le i+1,
        \qquad \deg Q_k=k+1.
\]
For a fixed degree $n$ we consider the linear Van Vleck polynomials $V$ for which
$\dq+V$ has a polynomial solution of degree $n$, and we form the spectral
polynomial $Sp_n$ whose zeros are the zeros of these Van Vleck polynomials.  The
main result is a finite-band determinant representation and the resulting
universality theorem: after normalization, all fixed power sums of the zeros of
$Sp_n$ have limits given by explicit constant-term formulae depending only on the
leading coefficient $Q_k$.  The lower coefficients of $\dq$ enter only lower order
correction terms.  Combining this with the localization theorem for Van Vleck
roots, we strengthen the usual germ-at-infinity conclusion to locally uniform
convergence of the normalized Cauchy transforms and logarithmic potentials on the
whole exterior of the convex hull of the zeros of $Q_k$.

We also prove a determinacy criterion: if the spectral roots are asymptotically
confined to a compact set with empty interior and connected complement, then the
finite-band moments determine the actual weak limit.  In particular, when the
zeros of $Q_k$ are collinear the root-counting measures of $Sp_n$ converge weakly
to a probability measure supported on the corresponding segment; this limit is
independent of all lower coefficients of $\dq$.  Finally, we prove holonomicity of
the exterior Cauchy transform and derive Picard--Fuchs equations for the WKB
periods, with an explicit third-order equation in the first non-classical case
$k=3$.  The paper ends with a precise mother-body conjecture for the genuinely
complex case, clearly separated from the unconditional results.
\end{abstract}

\maketitle

\section{Introduction and main results}
\label{sec:intro}

A \emph{generalized Lam\'e equation} \cite{WW} is a second order equation
\begin{equation}\label{eq:comLame}
        \left\{Q_2(z)\frac{d^2}{dz^2}+Q_1(z)\frac{d}{dz}+V(z)\right\}S(z)=0,
\end{equation}
where $Q_2$ is a polynomial of degree $l$ and $Q_1$ has degree at most
$l-1$.  Heine proved that, for a generic equation and for each positive
integer $n$, there are exactly $\binom{n+l-2}{n}$ polynomials $V$ for which
\eqref{eq:comLame} has a polynomial solution $S$ of degree $n$ \cite{He}.
The polynomial $V$ is a Van Vleck polynomial and $S$ is a Stieltjes polynomial.
The case $l=3$ is the classical Heun case.  If the roots of $Q_2$ are real,
then the corresponding Van Vleck zeros have a classical localization on the
real interval determined by the roots of $Q_2$; this is one of the origins of
the Heine--Stieltjes theory.

The aim of the present paper is to study the analogous spectral problem for
high-order operators of Fuchs index one,
\begin{equation}\label{eq:1}
        \dq=\sum_{i=1}^k Q_i(z)\frac{d^i}{dz^i},
        \qquad \deg Q_i\le i+1,
        \qquad \deg Q_k=k+1 .
\end{equation}
Following the classical terminology, we call such operators high-order Heun
operators.  For a fixed integer $n$ we consider all at most linear polynomials
$V(z)$ for which $\dq+V$ annihilates a polynomial of degree $n$.  For all
sufficiently large $n$ all these Van Vleck polynomials are linear, have the same
leading coefficient, and are $n+1$ in number counted with natural multiplicity;
see \cite{Sh} and \cite{BoSh}.  We write them in the form
\[
        V_{n,j}(z)=-\lambda_n(z-z_{n,j}),\qquad j=1,\ldots,n+1,
\]
and define the $n$-th spectral polynomial by
\[
        Sp_n(z)=\prod_{j=1}^{n+1}(z-z_{n,j}).
\]
The root-counting measure of $Sp_n$ will be denoted by $\mu_n$.

The classical second-order case was studied in \cite{ShT,STT}.  The main point
of the present paper is that, for arbitrary order $k$, the spectral polynomial
admits a finite-band determinant representation.  
 We shall use the following localization theorem from \cite{Sh,ShBook}.

\begin{theorem}[Localization of Van Vleck roots]\label{th:loc}
For every $\varepsilon>0$ there exists $n_\varepsilon$ such that, for every
$n>n_\varepsilon$ and every zero $z_{n,j}$ of $Sp_n$, the point $z_{n,j}$ lies in
the $\varepsilon$-neighbourhood of
\[
        K_Q=\Conv\{\zeta\in\mathbb C:Q_k(\zeta)=0\} .
\]
\end{theorem}

We now state the main results.  After an affine change of variable we may assume
that $Q_k$ is monic and has a zero at the origin,
\[
        Q_k(z)=z^{k+1}+a_kz^k+a_{k-1}z^{k-1}+\cdots+a_1z .
\]
The following result is the unconditional finite-band moment theorem.  

\begin{theorem}[Finite-band moments]\label{th:1}
Let $z_{n,1},\ldots,z_{n,n+1}$ be the zeros of $Sp_n$.  For every fixed
$m\ge1$ the limit
\[
        M_m(Q_k)=\lim_{n\to\infty}\frac1{n+1}\sum_{j=1}^{n+1} z_{n,j}^m
\]
exists and is given by
\begin{equation}\label{eq:moment-formula}
        M_m(Q_k)=\int_0^1
        \CT_w\left[
        \left((1-\tau^k)w^{-1}
        -\tau^k(a_k+a_{k-1}w+\cdots+a_1w^{k-1})\right)^m
        \right]d\tau .
\end{equation}
Here $\CT_w$ denotes the constant term in the Laurent polynomial in $w$.
Consequently every subsequential weak limit of $\{\mu_n\}$ has the same
holomorphic moments, and these moments depend only on $Q_k$.
The formula is affine-covariant: although it is written in a coordinate in which
one zero of $Q_k$ is placed at the origin, the resulting moments in the original
coordinate are independent of the chosen affine normalization.
\end{theorem}

Theorem~\ref{th:1} gives the exterior germ at infinity.  The next theorem
upgrades this germ to a global exterior statement.  Put
\[
        \Omega_Q=\mathbb C\setminus K_Q,
\]
and define the usual Cauchy transform and the logarithmic potential as 
\begin{align*}
        C_n(t)&=\int\frac{d\mu_n(z)}{t-z}
        =\frac1{n+1}\frac{Sp_n'(t)}{Sp_n(t)},\\
        U_n(t)&=\int\log|t-z|\,d\mu_n(z)
        =\frac1{n+1}\log|Sp_n(t)| .
\end{align*}

\begin{theorem}[Global exterior asymptotics]\label{th:2}
The functions $C_n$ converge locally uniformly in $\Omega_Q$ to an analytic
function $C_Q$ depending only on $Q_k$.  In a neighbourhood of infinity this
function is given by
\begin{equation}\label{eq:Cauchy-nu}
        C_Q(t)=\int_0^1 \frac{\partial}{\partial t}\log \Psi(t,\tau)\,d\tau,
\end{equation}
where $\Psi(t,\tau)$ is the branch 
of the algebraic equation
\begin{equation}\label{eq:rel-intro}
\Psi^k=(t+a_k\Theta)\Psi^{k-1}
+\Theta(1-\Theta)a_{k-1}\Psi^{k-2}
+\Theta(1-\Theta)^2a_{k-2}\Psi^{k-3}
+\cdots+
\Theta(1-\Theta)^{k-1}a_1,
\end{equation} 
satisfying $\Psi(t,\tau)=t+O(1)$ at infinity. Here  $\Theta=\tau^k$.  Moreover $U_n$ converges locally uniformly in $\Omega_Q$
to a harmonic function $U_Q^{\rm ext}$ normalized by
\[
        U_Q^{\rm ext}(t)=\log|t|+O(|t|^{-1}),\qquad t\to\infty,
\]
and satisfying
\[
        2\frac{\partial U_Q^{\rm ext}}{\partial t}=C_Q
        \qquad \text{in }\Omega_Q .
\]
Thus all subsequential weak limits of $\{\mu_n\}$ have the same exterior Cauchy
transform and the same exterior logarithmic potential on $\Omega_Q$.
\end{theorem}

The next theorem explains exactly when the moment information determines the
actual weak limit.  It is included because it gives a rigorous replacement for
the still conjectural tree statement in all cases where an independent
one-dimensional localization is available.

\begin{theorem}[Moment determinacy from one-dimensional localization]
\label{th:determinacy}
Assume that there exists a compact set $K\subset K_Q$ with empty interior and
connected complement such that, for every $\varepsilon>0$, all zeros of $Sp_n$
lie in the $\varepsilon$-neighbourhood of $K$ for all sufficiently large $n$.
Then $\mu_n$ converges weakly to a probability measure $\mu_Q$ supported on $K$.
The limit is uniquely characterized by
\[
        \int z^m\,d\mu_Q(z)=M_m(Q_k),\qquad m=0,1,2,\ldots,
\]
where $M_0(Q_k)=1$ and $M_m(Q_k)$ is given by \eqref{eq:moment-formula}.  In
particular the limit depends only on $Q_k$.
\end{theorem}

A direct and important consequence is the following full convergence theorem in
the collinear case.

\begin{theorem}[The collinear case]\label{th:collinear-convergence}
Assume that all zeros $\xi_0,\ldots,\xi_k$ of $Q_k$, counted with multiplicity,
lie on an affine line.  Then the root-counting measures $\mu_n$ converge weakly
to a probability measure $\mu_Q$ supported on the segment
\[
        K_Q=\Conv\{\xi_0,\ldots,\xi_k\}.
\]
This limit depends only on $Q_k$ and not on the lower coefficients of $\dq$.
If the line is the real line, then
\[
        \int x\,d\mu_Q(x)=\bar\xi,
        \qquad
        \bar\xi=\frac1{k+1}\sum_{j=0}^{k}\xi_j,
\]
and
\[
        \int (x-\bar\xi)^2\,d\mu_Q(x)
        =\frac{k}{2k+1}\,\frac1{k+1}\sum_{j=0}^{k}(\xi_j-\bar\xi)^2 .
\]
\end{theorem}

Theorem~\ref{th:collinear-convergence} proves the convergence part of the main
spectral-measure conjecture whenever the convex hull of the roots of $Q_k$ has
empty interior.  In the genuinely two-dimensional case holomorphic moments do not
usually determine a compactly supported positive measure, and this is where the
mother-body phenomenon enters.

\begin{definition}[Mother body]\label{def:mother-body}
Let $\rho$ be a compactly supported positive measure in the plane.  A positive
measure $\sigma$ is called an \emph{exterior mother body} for $\rho$ relative to
the unbounded component $\Omega_\infty$ of
$\mathbb C\setminus\operatorname{supp}\rho$ if
\begin{itemize}
\item $\operatorname{supp}\sigma$ is at most one-dimensional;
\item $U^\sigma(z)=U^\rho(z)$ for $z\in\Omega_\infty$, where
\[
        U^\rho(z)=\int\log|z-\zeta|\,d\rho(\zeta).
\]
\end{itemize}
Equivalently, $\rho$ and $\sigma$ have the same exterior Cauchy transform in
$\Omega_\infty$; see \cite{ShBook}.
\end{definition}

The finite-band construction also produces an averaged family of frozen
Toeplitz measures.  In the normalization used above, the relevant Laurent symbol
is
\begin{equation}\label{eq:symbol-intro}
        b_\tau(w)=(1-\tau^k)w^{-1}
        -\tau^k(a_k+a_{k-1}w+\cdots+a_1w^{k-1}).
\end{equation}
For fixed $\tau$, the Beraha--Kahane--Weiss/Schmidt--Spitzer theorem describes
the limiting zero distribution of the constant-coefficient recurrence with symbol
$b_\tau$; see \cite{BKW,SS,BBS2}.  Averaging these frozen measures over
$0\le\tau\le1$ gives a measure whose exterior potential is the one appearing in
Theorem~\ref{th:2}.  For complex $Q_k$ this averaged finite-band measure typically has 
 two-dimensional support, while numerical computations of $Sp_n$ show a
thin tree-like carrier.  The expected limiting measure is therefore not the
averaged measure itself, but a mother body for its exterior potential.

\medskip
This leads to the following conjecture strengthening the above unconditional
results.

\begin{conjecture}[Spectral tree / mother-body conjecture]\label{conj:main}
For a high-order Heun operator \eqref{eq:1} with simple roots of $Q_k$, the
measures $\mu_n$ converge weakly to a probability measure $\mu_Q$ depending only
on $Q_k$.  Its support is a finite planar tree contained in $K_Q$, and the leaves
of the tree are exactly the roots of $Q_k$.  Moreover $\mu_Q$ is a positive
mother body for the exterior potential $U_Q^{\rm ext}$ of Theorem~\ref{th:2}.
\end{conjecture}

\begin{figure}[H]
\begin{center}
\IfFileExists{compl.pdf}{\includegraphics[width=0.45\linewidth]{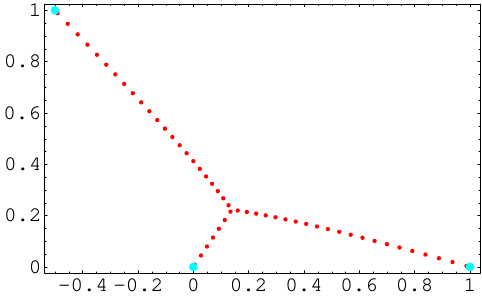}}{\fbox{\parbox{0.75\linewidth}{Figure file \texttt{compl.pdf} not found.}}}
\end{center}
\caption{The roots of the spectral polynomial $Sp_{50}(t)$ for the operator
$\dq=Q(z)d^2/dz^2$ with $Q(z)=z(z-1)(z+\frac12-i)$.}
\label{fig2}
\end{figure}

\begin{figure}[H]
\begin{center}
\IfFileExists{Q5d4_bright_manual.pdf}{\includegraphics[width=0.65\linewidth]{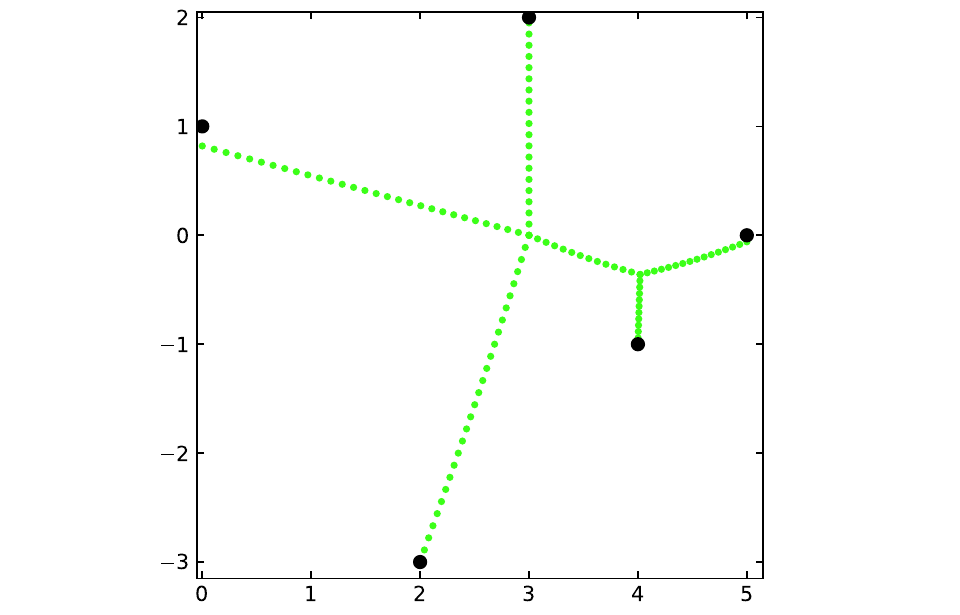}}{\fbox{\parbox{0.75\linewidth}{Figure file \texttt{Q5d4\_bright\_manual.pdf} not found.}}}
\end{center}
\caption{The roots of the spectral polynomial $Sp_{150}(z)$ for the operator
$\dq=Q(z)d^4/dz^4$ with
$Q(z)=(z-5)(z-4+i)(z-2+3i)(z-i)(z-3-2i)$.}
\label{fig1}
\end{figure}

Figures~\ref{fig2}--\ref{fig:nu-k3} illustrate this distinction between the
proved exterior potential and the conjectural one-dimensional carrier.

\medskip
\noindent\emph{Explanation for Figures~\ref{fig2}--\ref{fig1}.}  Larger dots
are the roots of $Q(z)$ and smaller dots are roots of the corresponding spectral
polynomial.

Finally, we prove that the exterior Cauchy transform is holonomic.  More
precisely, creative telescoping applied to the algebraic integrand in
\eqref{eq:Cauchy-nu} gives a linear differential equation with polynomial
coefficients for $C_Q$, up to explicitly computable endpoint terms.  In parallel,
the WKB periods satisfy a simple Picard--Fuchs equation; in the first
non-classical case $k=3$ this equation becomes
\[
        81Q(t)I'''(t)+162Q'(t)I''(t)+90Q''(t)I'(t)+10Q'''(t)I(t)=0.
\]
This WKB equation is rigorous as a period equation, but its identification with
the spectral tree is left as a conjectural program.

The paper is organized as follows.  Section~\ref{sec:finite-band-kva} proves the
finite-band trace principle.  Section~\ref{Proofs} gives the determinant
representation and proves Theorems~\ref{th:1}, \ref{th:2},
\ref{th:determinacy}, and \ref{th:collinear-convergence}.  Section~\ref{sec:pf-k3}
proves holonomicity and derives the Picard--Fuchs equations.  Section~\ref{sec:mother-body}
explains the mother-body interpretation, and Section~\ref{sec:wkb} formulates the
WKB route to the conjectural spectral tree.  The final section lists open
problems.

In the first non-classical case $k=3$, the quartic leading coefficient has four
zeros.  For
\[
        Q(z)=z(z-1)(z-i)(z-1-i)
\]
one obtains four natural averaged finite-band measures according to which zero is
chosen as the distinguished origin.  Figure~\ref{fig:nu-k3} compares these
auxiliary two-dimensional supports with the high-precision roots of $Sp_{40}$ for
$Q(z)d^3/dz^3$.

\section {A finite-band Kuijlaars--Van Assche principle}
\label{sec:finite-band-kva}

We shall use only the exterior-potential part of the Kuijlaars--Van Assche method.  This distinction is important in the present complex setting.  In the real Jacobi case Theorem~1.4 of \cite{KvA} gives weak convergence of zero-counting measures.  For complex recurrence coefficients one should not expect such a statement without additional information on the support.  Appendix~II of \cite{ShTquartic} gives the appropriate complex version in the tridiagonal case: one obtains the limiting logarithmic potential near infinity, and if a weak limit exists then it is equipotential with the averaged frozen measure outside the union of the supports.  Similar statements can be found in \cite{KMS}. The finite-band form needed below is the following elementary variant.

\begin{figure}[H]
\begin{center}
\IfFileExists{nu_k3_four_measures_corrected.pdf}{\includegraphics[width=0.96\linewidth]{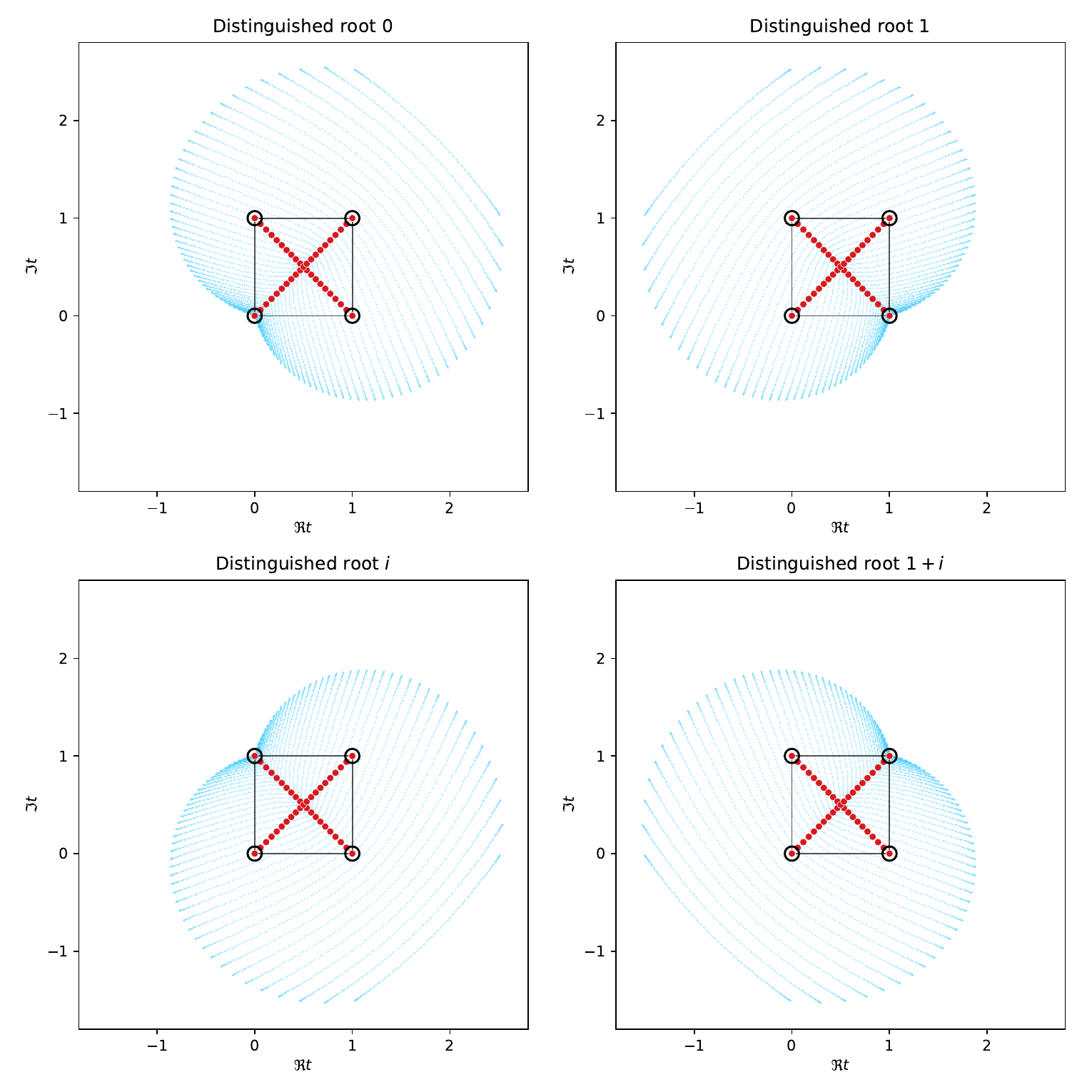}}{\fbox{\parbox{0.85\linewidth}{Figure file \texttt{nu\_k3\_four\_measures\_corrected.pdf} not found.}}}
\end{center}
\caption{Numerical approximations of the four auxiliary averaged measures for
$Q(z)=z(z-1)(z-i)(z-1-i)$.  The bright blue sets are sampled from the
frozen equimodular condition, the red points are high-precision roots of
$Sp_{40}$ for $Q(z)d^3/dz^3$, and the black circles are the zeros of $Q$.}
\label{fig:nu-k3}
\end{figure}

Let $p,q$ be fixed non-negative integers and let $B_N$ be a sequence of $(N+1)\times (N+1)$ matrices, indexed by $0,1,\ldots,N$, with uniformly bounded entries and with fixed bandwidth,
\[
        (B_N)_{ij}=0\qquad \text{if } i-j\notin[-q,p].
\]
Assume that for every diagonal $\ell\in[-q,p]$ there exists a continuous function $b_\ell:[0,1]\to\mathbb C$ such that
\[
        (B_N)_{j,j+\ell}\longrightarrow b_\ell(\tau)
        \qquad \text{whenever } j/N\to\tau,
\]
uniformly with respect to admissible indices away from the boundary.  Put
\begin{equation}\label{eq:general-symbol}
        b_\tau(w)=\sum_{\ell=-q}^{p} b_\ell(\tau)w^\ell .
\end{equation}
For a Laurent polynomial $L(w)$, denote by $\CT_w L$ its constant term.

\begin{theorem}[Finite-band theorem]
\label{th:mainrec}
Under the assumptions above, for every fixed $m\ge 1$ one has
\begin{equation}\label{eq:general-moments}
        \lim_{N\to\infty}\frac{1}{N+1}\operatorname{tr} B_N^m
        =\int_0^1 \CT_w\, b_\tau(w)^m\,d\tau .
\end{equation}
Consequently, if
\[
        P_N(t)=\det(tI-B_N),
\]
then for $|t|$ sufficiently large,
\begin{equation}\label{eq:general-logdet}
\lim_{N\to\infty}\frac{1}{N+1}\log P_N(t)
=
\log t-
\sum_{m\ge 1}\frac{1}{m t^m}
\int_0^1 \CT_w\, b_\tau(w)^m\,d\tau,
\end{equation}
where the principal branch is chosen near infinity.  The convergence is locally uniform in a neighbourhood of infinity.  Equivalently,
\begin{equation}\label{eq:general-cauchy}
\lim_{N\to\infty}\frac{1}{N+1}\frac{P_N'(t)}{P_N(t)}
=
\frac1t+
\sum_{m\ge 1}\frac{1}{t^{m+1}}
\int_0^1 \CT_w\, b_\tau(w)^m\,d\tau .
\end{equation}
In particular, every subsequential weak limit of the zero-counting measures of $P_N$ has the same logarithmic potential near infinity, and this exterior potential is determined only by the limiting symbol \eqref{eq:general-symbol}.
\end{theorem}

\begin{proof}
The trace of $B_N^m$ is a sum over closed paths of length $m$ in the directed graph of the band matrix.  Paths meeting the first or last $O(m)$ rows contribute only $O(1)$ to the trace, hence disappear after division by $N+1$.  For an interior starting index $j$ with $j/N\to\tau$, the contribution of closed paths is exactly the constant term of the $m$-th power of the local Laurent symbol, up to an error tending to zero:
\[
        \sum_{\ell_1+\cdots+\ell_m=0}
        b_{\ell_1}(\tau)b_{\ell_2}(\tau)\cdots b_{\ell_m}(\tau)
        =\CT_w\, b_\tau(w)^m .
\]
Summing over $j$ gives a Riemann sum and proves \eqref{eq:general-moments}.

Choose $R$ so large that $\|B_N\|\le R$ for all $N$.  For $|t|>R$,
\[
        \log\det(tI-B_N)
        =(N+1)\log t+
        \operatorname{tr}\log(I-t^{-1}B_N)
        =(N+1)\log t-
        \sum_{m\ge 1}\frac{1}{m t^m}\operatorname{tr}B_N^m .
\]
The series is normally convergent on compact subsets of $|t|>R$, so the limit may be passed through the sum.  This gives \eqref{eq:general-logdet}, and differentiation gives \eqref{eq:general-cauchy}.
\end{proof}

\begin{remark}
For a fixed value of $\tau$, the symbol $b_\tau$ is the symbol of a finite-band Toeplitz matrix.  The corresponding limiting zero distribution is described by the usual Beraha--Kahane--Weiss/Schmidt--Spitzer equimodular condition.  The theorem above says that a slowly varying band matrix has, near infinity, the averaged exterior potential of these frozen Toeplitz limits.  In the real tridiagonal situation this is precisely the content of \cite{KvA}; in the complex tridiagonal situation one should use the more cautious potential-theoretic formulation of Appendix~II in \cite{ShTquartic}.  This is the point needed later: the averaged finite-band measure determines an exterior potential, but the actual weak limit may be a different measure with the same exterior potential.
\end{remark}

\section{Proofs} \label{Proofs}

We first explain the finite-band determinant representation.  After an affine change of variable we may assume that $Q_k$ is monic and has a zero at the origin,
\[
Q_k(z)=z^{k+1}+a_kz^k+a_{k-1}z^{k-1}+\cdots+a_1z .
\]
Write
\[
        Q_i(z)=\sum_{j=0}^{i+1}q_{i,j}z^j,
\]
and set
\[
        \lambda_n=\sum_{i=1}^k q_{i,i+1}\,n(n-1)\cdots(n-i+1).
\]

\begin{lemma}[No spurious lower-degree kernel vectors]
\label{lem:no-spurious}
Let
\[
        \lambda(s)=\sum_{i=1}^{k}q_{i,i+1}\,s(s-1)\cdots(s-i+1),
\]
so that $\lambda_n=\lambda(n)$.  There exists $n_0$ such that
\[
        \lambda(n)\ne\lambda(s),\qquad 0\le s<n,\quad n\ge n_0 .
\]
Consequently, for $n\ge n_0$, every non-zero vector in the kernel of
\[
        \dq-\lambda_n(z-t):\operatorname{Pol}_{\le n}\longrightarrow
        \operatorname{Pol}_{\le n}
\]
has exact degree $n$.  Hence
\[
        Sp_n(t)=\det(tI-Z+\lambda_n^{-1}\dq)
\]
with no extra factors coming from polynomial solutions of lower degree.
\end{lemma}

\begin{proof}
The polynomial $\lambda(s)$ has degree $k$ and leading coefficient
$q_{k,k+1}\ne0$.  Suppose that $\lambda(n_\nu)=\lambda(s_\nu)$ for a sequence
$n_\nu\to\infty$ and integers $0\le s_\nu<n_\nu$.  Passing to a subsequence,
write $s_\nu/n_\nu\to\rho\in[0,1]$.  If $\rho<1$, division by $n_\nu^k$ gives
$q_{k,k+1}(1-\rho^k)=0$, a contradiction.  Hence $\rho=1$.  Put
$h_\nu=n_\nu-s_\nu$.  Then $h_\nu\ge1$ and $h_\nu=o(n_\nu)$.  Expanding
$\lambda(n_\nu)-\lambda(n_\nu-h_\nu)$ gives
\[
        \lambda(n_\nu)-\lambda(n_\nu-h_\nu)
        =kq_{k,k+1}h_\nu n_\nu^{k-1}
        +O(h_\nu^2 n_\nu^{k-2})+O(h_\nu n_\nu^{k-2}).
\]
After division by $h_\nu n_\nu^{k-1}$ this tends to $kq_{k,k+1}\ne0$, again a
contradiction.  Thus the numbers $\lambda(n)$ are eventually separated from all
previous values $\lambda(s)$.

Now let $S$ be a non-zero kernel vector of degree $s\le n$ with leading
coefficient $c_s$.  If $s<n$, then the coefficient of $z^{s+1}$ in
$(\dq-\lambda_n z+\lambda_n t)S$ equals $(\lambda(s)-\lambda(n))c_s$, which is
non-zero by the first part.  Therefore $s=n$.  The determinant therefore records
exactly the degree-$n$ Van Vleck parameters, counted with their natural algebraic
multiplicities.
\end{proof}
Since $Q_k$ is monic, $\lambda_n=n(n-1)\cdots(n-k+1)+O(n^{k-1})$ and is non-zero for all sufficiently large $n$.  For a Stieltjes polynomial of degree $n$, the coefficient of $z^{n+1}$ in
\[
\left(\dq+V\right)S
\]
must vanish. Hence all admissible linear Van Vleck polynomials have the same leading coefficient. With the sign convention used below we write
\[
V_n(z)=-\lambda_n(z-t).
\]
The equation
\[
(\dq-\lambda_n(z-t))S=0,
\qquad \deg S\le n,
\]
is equivalent to degeneracy of the induced operator on $Pol_{\le n}$. Dividing by $\lambda_n$, this induced operator is represented in the monomial basis by a finite-band matrix
\[
M_n(t)=tI-Z+\lambda_n^{-1}\dq,
\]
where $Z$ is the matrix of multiplication by $z$ truncated to $Pol_{\le n}$. Therefore
\[
Sp_n(t)=\det M_n(t)
\]
exactly, since both sides are monic polynomials in $t$ with the same zeros
counted with the same multiplicities.

By construction $M_n(t)$ has exactly $k+1$ possibly non-zero diagonals: the first subdiagonal coming from $-Z$, the main diagonal containing $t$, and $k-1$ upper diagonals coming from $Q_k(z)d^k/dz^k$ and the lower order terms. Let
\[
Sp_{n,j}(t)=\det M_{n,j}(t),\qquad 1\le j\le n+1,
\]
where $M_{n,j}$ is the leading $j\times j$ principal minor of $M_n(t)$, and set $Sp_{n,0}=1$. Since $M_n$ is banded, the sequence $Sp_{n,j}$ satisfies a finite recurrence of order $k$ in $j$.

The lower-order summands of $\dq$ do not contribute to the limiting local symbol, and hence do not contribute to the limiting recurrence for principal minors. Indeed, the contribution of
\[
Q_i(z)\frac{d^i}{dz^i},\qquad i<k,
\]
after division by $\lambda_n$, is $O(n^{i-k})$ uniformly on the relevant diagonals. Hence these terms vanish in the local limit $j/n\to\tau$.

\begin{theorem}\label{th:rec}
If
\[
Q_k(z)=z^{k+1}+a_kz^k+a_{k-1}z^{k-1}+\cdots+a_1z,
\]
then, for $j/n\to\tau$, the limiting recurrence for the principal minors has characteristic equation
\begin{equation}\label{eq:rel}
\Psi^k=(z+a_k\Theta)\Psi^{k-1}
+\Theta(1-\Theta)a_{k-1}\Psi^{k-2}
+\Theta(1-\Theta)^2a_{k-2}\Psi^{k-3}
+\cdots+
\Theta(1-\Theta)^{k-1}a_1,
\end{equation}
where $\Theta=\tau^k$.
\end{theorem}

\begin{proof}
It is cleaner to read the recurrence from the local matrix symbol of
$B_n=Z-\lambda_n^{-1}\dq$.  For an interior monomial $z^j$ with
$j/n\to\tau$, the two terms mapping $z^j$ to $z^{j+1}$ are $Z$ and
$z^{k+1}d^k/dz^k$.  Their combined coefficient tends to
$1-\tau^k=1-\Theta$.  The term
$a_{k-r}z^{k+1-r}d^k/dz^k$, $r=1,\ldots,k$, maps $z^j$ to $z^{j+1-r}$
and contributes, after division by $\lambda_n$, the coefficient
$-\Theta a_{k-r}$ to $B_n$.  Therefore the frozen local Laurent symbol is
\[
        b_\tau(w)=(1-\Theta)w^{-1}
        -\Theta(a_k+a_{k-1}w+\cdots+a_1w^{k-1}).
\]
The characteristic equation for the corresponding Toeplitz recurrence is
$t=b_\tau(w)$.  Setting $\Psi=(1-\Theta)/w$ and clearing denominators gives
exactly \eqref{eq:rel}, with $z$ there playing the role of the spectral
parameter $t$.  This proves the assertion.
\end{proof}

Equivalently, the same frozen problem is encoded by the Laurent symbol
\begin{equation}\label{eq:symbol}
        b_\tau(w)=(1-\tau^k)w^{-1}-\tau^k(a_k+a_{k-1}w+\cdots+a_1w^{k-1}).
\end{equation}
Indeed, setting $\Psi=(1-\tau^k)/w$ transforms the equation $t=b_\tau(w)$ into \eqref{eq:rel}.  This elementary observation is useful because the moments of the spectral roots are constant terms of powers of $b_\tau$.

For each fixed $\tau$, define $\nu_{Q_k}(\tau)$ to be the limiting zero-counting measure of the constant-coefficient recurrence whose characteristic equation is \eqref{eq:rel}. Define
\[
\nu_\dq=\int_0^1 \nu_{Q_k}(\tau)d\tau .
\]
The notation emphasizes the fact that this averaged measure depends only on $Q_k$.

\begin{proof}[Proof of Theorem~\ref{th:1}]
Put
\[
        B_n=Z-\lambda_n^{-1}\dq,
        \qquad M_n(t)=tI-B_n .
\]
Then $Sp_n(t)=\det(tI-B_n)$, and hence
\[
        \frac{1}{n+1}\sum_{j=1}^{n+1}z_{n,j}^m
        =\frac{1}{n+1}\operatorname{tr} B_n^m .
\]
For fixed $m$, the trace of $B_n^m$ is a sum over closed paths of length $m$ in the finite band of $B_n$.  Boundary paths involving the first or last $O(m)$ rows contribute $O(1)$ to the trace and therefore disappear after division by $n+1$.

For an interior index $j$ with $j/n\to\tau$, the non-zero diagonals of $B_n$ converge to the coefficients of the Laurent symbol \eqref{eq:symbol}.  The lower order parts of $\dq$ vanish uniformly in this local limit, because the contribution of $Q_i(z)d^i/dz^i$, $i<k$, is $O(n^{i-k})$ after division by $\lambda_n$.  Thus the contribution of all closed paths based at such an index tends to the constant term
\[
        \CT_w\, b_\tau(w)^m .
\]
Summing over $j=0,1,\ldots,n$ gives a Riemann sum, and we obtain exactly \eqref{eq:moment-formula}.  The formula contains only the coefficients of $Q_k$, proving the asserted independence of the lower coefficients.
\end{proof}

\begin{corollary}[First moment consequences]\label{cor:first-moments}
In the coordinate in which
$Q_k(z)=z^{k+1}+a_kz^k+a_{k-1}z^{k-1}+\cdots+a_1z$, every subsequential weak limit $\mu$ of the root-counting measures of $Sp_n$ satisfies
\[
        \int z\,d\mu(z)=-\frac{a_k}{k+1},
\]
so the barycenter of the limiting Van Vleck roots is the arithmetic mean of the zeros of $Q_k$.  Moreover
\[
        \int z^2\,d\mu(z)=
        \frac{a_k^2}{2k+1}
        -\frac{2k}{(k+1)(2k+1)}a_{k-1},
\]
and
\[
        \int z^3\,d\mu(z)=
        -\frac{a_k^3}{3k+1}
        +\frac{6k}{(2k+1)(3k+1)}a_ka_{k-1}
        -\frac{6k^2}{(k+1)(2k+1)(3k+1)}a_{k-2}.
\]
\end{corollary}

\begin{proof}
These identities are obtained by taking the constant terms of the first three powers of the symbol \eqref{eq:symbol} and integrating in $\tau$.
\end{proof}

\begin{proof}[Proof of Theorem~\ref{th:2}]
By Theorem~\ref{th:1}, the normalized Cauchy transforms converge as germs at
infinity to the Laurent series
\[
        \frac1t+\sum_{m\ge1}\frac{M_m(Q_k)}{t^{m+1}}.
\]
For the frozen recurrence with parameter $\tau$, let $\Psi(t,\tau)$ be the
branch of \eqref{eq:rel} satisfying $\Psi(t,\tau)=t+O(1)$ at infinity.  The
standard ratio-asymptotic formula for finite-term recurrences gives
\[
        \lim_{m\to\infty}\frac1m\frac{d}{dt}\log P_{m,\tau}(t)
        =\frac{\partial}{\partial t}\log\Psi(t,\tau)
\]
locally uniformly near infinity.  Averaging over $0\le\tau\le1$ and using the
symbol identity \eqref{eq:symbol} identifies the germ with
\[
        C_Q(t)=\int_0^1\frac{\partial}{\partial t}\log\Psi(t,\tau)\,d\tau .
\]
This proves the asserted formula near infinity.

It remains to extend the convergence from a neighbourhood of infinity to the
whole exterior $\Omega_Q$.  Let $L\Subset\Omega_Q$ and put
$\delta=\operatorname{dist}(L,K_Q)>0$.  Choose
$0<\varepsilon<\delta/2$.  By Theorem~\ref{th:loc}, for all sufficiently large
$n$ all zeros of $Sp_n$ lie in the $\varepsilon$-neighbourhood of $K_Q$.  Hence
for $t\in L$,
\[
        |C_n(t)|\le \frac{2}{\delta} .
\]
The family $\{C_n\}$ is therefore normal on $\Omega_Q$.  Every subsequential
locally uniform limit agrees with the already identified germ near infinity, and
hence agrees with its analytic continuation throughout the connected domain
$\Omega_Q$.  Thus the whole sequence $C_n$ converges locally uniformly on
$\Omega_Q$.

The potentials $U_n$ are harmonic on $L$ for all sufficiently large $n$ and are
locally uniformly bounded there.  Hence they form a normal family of harmonic
functions.  Near infinity their limit is fixed by the expansion
$U_n(t)=\log|t|+O(|t|^{-1})$ and by the convergence of the Cauchy transforms.
On $\Omega_Q$ the gradient is determined by the limit $C_Q$, and the
normalization at infinity fixes the additive constant.  Therefore $U_n$ converges
locally uniformly to the harmonic function $U_Q^{\rm ext}$ stated in the theorem.
\end{proof}

\begin{proof}[Proof of Theorem~\ref{th:determinacy}]
The sequence $\{\mu_n\}$ is tight by Theorem~\ref{th:loc}.  Let $\mu$ be any
subsequential weak limit.  The assumed localization near $K$ implies
$\operatorname{supp}\mu\subset K$.  By Theorem~\ref{th:1},
\[
        \int z^m\,d\mu(z)=M_m(Q_k),\qquad m=0,1,2,\ldots .
\]
If $\mu$ and $\nu$ are two subsequential limits, then they have the same
integrals against all polynomials in $z$.  Since $K$ has empty interior and
connected complement, Mergelyan's theorem implies that polynomials are uniformly
dense in $C(K)$.  Hence $\int f\,d\mu=\int f\,d\nu$ for every continuous function
$f$ on $K$, and therefore $\mu=\nu$.  All subsequential limits coincide, which is
equivalent to weak convergence of the whole sequence.  The limit depends only on
$Q_k$ because its moments are given by \eqref{eq:moment-formula}.
\end{proof}

\begin{proof}[Proof of Theorem~\ref{th:collinear-convergence}]
If all zeros of $Q_k$ are collinear, then $K_Q$ is a compact segment.  Theorem
\ref{th:loc} gives the localization hypothesis of Theorem~\ref{th:determinacy}
with $K=K_Q$, and therefore $\mu_n$ converges weakly to a probability measure
$\mu_Q$ supported on $K_Q$ and depending only on $Q_k$.

Assume now that the line is the real line and write
\[
        Q_k(z)=\prod_{j=0}^k(z-\xi_j)
        =z^{k+1}+a_kz^k+a_{k-1}z^{k-1}+\cdots .
\]
Then $a_k=-\sum_j\xi_j$ and $a_{k-1}=\sum_{i<j}\xi_i\xi_j$.  Corollary
\ref{cor:first-moments} gives
\[
        \int x\,d\mu_Q(x)=-\frac{a_k}{k+1}=\bar\xi.
\]
Substitution of the same elementary symmetric functions into the second moment
formula in Corollary~\ref{cor:first-moments} yields
\[
        \int (x-\bar\xi)^2\,d\mu_Q(x)
        =\frac{k}{2k+1}\left(\frac1{k+1}\sum_{j=0}^k\xi_j^2-\bar\xi^2\right),
\]
which is the announced formula.
\end{proof}

\begin{proposition}[Complete expansion of fixed spectral moments]
\label{prop:moment-expansion}
For every fixed $m\ge1$ the normalized power sum
\[
        \frac1{n+1}\sum_{j=1}^{n+1}z_{n,j}^m
        =\frac1{n+1}\operatorname{tr}B_n^m
\]
has a complete asymptotic expansion in descending powers of $n$,
\[
        \frac1{n+1}\operatorname{tr}B_n^m
        \sim \sum_{r\ge0}\frac{M_{m,r}}{n^r},\qquad n\to\infty .
\]
The leading coefficient is $M_{m,0}=M_m(Q_k)$.  Each coefficient $M_{m,r}$ is
obtained by a finite algorithm from the coefficients of $Q_1,\ldots,Q_k$; the
lower coefficients of $\dq$ first occur in correction terms $r\ge1$.
\end{proposition}

\begin{proof}
The entries of $B_n=Z-\lambda_n^{-1}\dq$ on each fixed diagonal are finite sums
of terms of the form
\[
        \lambda_n^{-1}q_{i,j}\,s(s-1)\cdots(s-i+1),
\]
where $s$ is the column index.  Since $\lambda_n$ is a degree $k$ polynomial in
$n$ with non-zero leading coefficient, these entries have, uniformly for
$s/n$ in compact subintervals of $[0,1]$, complete expansions in powers of
$n^{-1}$ whose coefficients are polynomials in $s/n$.  For fixed $m$, the trace
of $B_n^m$ is a finite sum over closed band paths of length $m$.  Away from the
first and last $O(m)$ rows, the contribution of each path therefore has a
complete expansion as a smooth function of $s/n$.  Euler--Maclaurin summation
gives a complete expansion for the sum over $s$, while the finitely many boundary
rows contribute only their own finite expansions.  Division by $n+1$ gives the
claim.  The leading term is the constant-term/Riemann-sum expression of
Theorem~\ref{th:1}.
\end{proof}

\noindent 
{\it Examples.} (1) Set $k=2$ and $Q=z^3+a_2z^2+a_1z$. Then
\[
\Psi^2=(z+a_2\tau^2)\Psi+(1-\tau^2)\tau^2a_1 .
\]
With the branch chosen by $\Psi=z+O(1)$ at infinity,
\[
\C(z)=\int_0^1\frac{\Psi'_z}{\Psi}\,d\tau
      =\int_0^1\frac{d\tau}{2\Psi-(z+a_2\tau^2)}
      =\int_0^1 \frac{d\tau}
      {\sqrt{(z+a_2\tau^2)^2+4(1-\tau^2)\tau^2a_1}},
\]
where the square root is chosen to be $z+O(1)$ at infinity.

\medskip \noindent
(2) Set $k=3$ and $Q=z^4+a_3z^3+a_2z^2+a_1z$. Then
\[
\Psi^3=(z+a_3\tau^3)\Psi^2+(1-\tau^3)\tau^3a_2\Psi+(1-\tau^3)^2\tau^3a_1.
\]
Writing $A=z+a_3\tau^3$, $B=(1-\tau^3)\tau^3a_2$ and
$D=(1-\tau^3)^2\tau^3a_1$, we get
\[
\C(z)=\int_0^1\frac{\Psi'_z}{\Psi}\,d\tau
     =\int_0^1\frac{\Psi\,d\tau}{3\Psi^2-2A\Psi-B}
     =\int_0^1\frac{\Psi^2\,d\tau}{A\Psi^2+2B\Psi+3D}.
\]

\section{Holonomicity, Picard--Fuchs equations, and the case $k=3$}
\label{sec:pf-k3}

This section separates two related but logically different statements.  First, the
finite-band Cauchy transform \eqref{eq:Cauchy-nu} is holonomic in the spectral
parameter.  This is a rigorous consequence of creative telescoping.  Second, the
WKB periods satisfy a relatively simple Picard--Fuchs equation.  The WKB
Picard--Fuchs equation is expected to describe the one-dimensional mother body
selected by the spectral polynomials, but this identification is still conjectural.

\begin{proposition}[Holonomicity of the finite-band Cauchy transform]
\label{prop:finiteband-holonomic}
For every $k$ and every polynomial $Q_k$ as above, the germ $C_Q(t)$ defined by
\eqref{eq:Cauchy-nu} is holonomic.  Equivalently, there exist an integer $N$ and
polynomials $A_0(t),\ldots,A_N(t)$, not all zero, such that
\begin{equation}\label{eq:finiteband-ode-general}
        A_N(t)C_Q^{(N)}(t)+\cdots+A_1(t)C_Q'(t)+A_0(t)C_Q(t)=B(t),
\end{equation}
where $B(t)$ is an explicitly computable algebraic endpoint term.  The operator
and the endpoint term can be obtained algorithmically by creative telescoping from
the algebraic integrand in \eqref{eq:Cauchy-nu}.
\end{proposition}

\begin{proof}
The function $\Psi(t,\tau)$ is algebraic over $\mathbb C(t,\tau)$, because it is
defined by the polynomial equation \eqref{eq:rel}.  Therefore
\[
        G(t,\tau)=\frac{\partial}{\partial t}\log\Psi(t,\tau)
\]
is algebraic, and in particular holonomic, as a function of the two variables
$(t,\tau)$.  The creative-telescoping theorem for holonomic functions, see e.g.
\cite{Zeilberger}, gives a non-zero operator
\[
        L(t,\partial_t)=A_N(t)\partial_t^N+\cdots+A_0(t)
\]
and a holonomic certificate $H(t,\tau)$ such that
\[
        L(t,\partial_t)G(t,\tau)=\frac{\partial}{\partial\tau}H(t,\tau).
\]
Integrating over $0\leq \tau\leq 1$ gives \eqref{eq:finiteband-ode-general}, with
$B(t)=H(t,1)-H(t,0)$.  This proves holonomicity of the finite-band exterior
Cauchy transform.  The argument is constructive, although the order and the
coefficients of the telescoper can grow quickly with $k$.
\end{proof}

The preceding proposition gives an algorithmic linear
ODE for the finite-band exterior Cauchy transform.  What it does not give, at
least not in a simple closed form, is the geometric tree carrying the conjectural
mother-body measure.  For that purpose the WKB Picard--Fuchs equation seems more
transparent.

\medskip
The following explicit elimination algorithm is the practical version of Proposition~\ref{prop:finiteband-holonomic}.  It applies without any conceptual change for every order $k$, in particular for $k>4$.

\begin{enumerate}
\item Choose one zero of $Q_k$ as the distinguished origin and write
\[
        Q_k(z)=z^{k+1}+a_kz^k+\cdots+a_1z .
\]
Put $\theta=\tau^k$ and form the algebraic equation
\begin{equation}\label{eq:alg-psi}
        F(\Psi,t,\tau)=
        \Psi^k-(t+a_k\theta)\Psi^{k-1}
        -\sum_{j=1}^{k-1}(1-\theta)^j\theta a_{k-j}\Psi^{k-1-j}=0 .
\end{equation}
The branch relevant at infinity is characterized by $\Psi=t+O(1)$.

\smallskip
\item Form the algebraic integrand
\begin{equation}\label{eq:alg-integrand}
        G(t,\tau)=\partial_t\log\Psi(t,\tau)
        =-\frac{F_t(\Psi,t,\tau)}{\Psi F_\Psi(\Psi,t,\tau)}
\tag{*}
\end{equation}
with the sign convention determined by differentiating $F(\Psi,t,\tau)=0$.  Equivalently, use $\Psi_t=-F_t/F_\Psi$ and $G=\Psi_t/\Psi$.

\smallskip
\item Work in the finite algebraic extension
\[
        \mathbb C(t,\tau)[\Psi]/(F).
\]
Every derivative $\partial_t^mG$ and $\partial_\tau H$ is reduced modulo $F$ to a polynomial in $\Psi$ of degree at most $k-1$ with coefficients in $\mathbb C(t,\tau)$.

\smallskip
\item Search for a telescoping identity
\begin{equation}\label{eq:alg-telescope}
        \sum_{m=0}^N A_m(t)\partial_t^mG(t,\tau)=\partial_\tau H(t,\tau)
\end{equation}
inside this finite-dimensional module.  After clearing denominators, this is a finite homogeneous linear system for the coefficients of the unknown polynomials $A_m(t)$ and of the certificate $H$.  Increasing $N$ and the allowed degree bounds eventually succeeds by holonomicity.

\item Integrate \eqref{eq:alg-telescope} from $\tau=0$ to $\tau=1$.  Since
\[
        C_Q(t)=\int_0^1G(t,\tau)d\tau,
\]
one obtains
\begin{equation}\label{eq:alg-final-ode}
        \sum_{m=0}^N A_m(t) C_Q^{(m)}(t)=H(t,1)-H(t,0).
\end{equation}
The right-hand side is the endpoint contribution.  It is usually much simpler than the intermediate certificate.
\end{enumerate}

The algorithm is not meant to produce a pleasant closed formula for large $k$; it is meant to prove existence and to make the differential equation computable.  On the other hand, the Picard--Fuchs equation below is geometrically more transparent, but it governs WKB periods rather than, at present, the finite-band Cauchy transform itself.

\medskip
Let
\[
        Q(t)=t^{k+1}+a_kt^k+\cdots+a_0
\]
be a polynomial of degree $k+1$, and consider the WKB curve
\begin{equation}
\label{eq:pf-curve}
        \mathcal X_t:\qquad y^k=\frac{z-t}{Q(z)}.
\end{equation}
The differential whose periods are obtained by differentiating the WKB action
\eqref{eq:wkb-differential} with respect to $t$ is, up to a non-zero constant,
\begin{equation}
\label{eq:eta-general}
        \eta_t=(z-t)^{1/k-1}Q(z)^{-1/k}\,dz .
\end{equation}

\begin{proposition}[Picard--Fuchs equation for WKB periods]
\label{prop:general-pf}
Let $\gamma$ be a locally constant cycle on \eqref{eq:pf-curve}, and set
\[
        I_\gamma(t)=\int_\gamma \eta_t .
\]
Put $\alpha=(k-1)/k$.  Then $I_\gamma$ satisfies the order-$k$ linear equation
\begin{equation}
\label{eq:general-pf}
        \mathcal L_k I_\gamma=0,
        \qquad
        \mathcal L_k=
        \sum_{j=0}^{k}
        \frac{k+1-j}{j!(\alpha)_{k-j}}
        Q^{(j)}(t)\frac{d^{k-j}}{dt^{k-j}},
\end{equation}
where $(\alpha)_m$ is the Pochhammer symbol.
\end{proposition}

\begin{proof}
Write $s=z-t$ and
\[
        F=s^{1/k-1}Q(z)^{-1/k}.
\]
Then
\[
        \frac{d^r}{dt^r}F=(\alpha)_r s^{-r}F,
        \qquad \alpha=\frac{k-1}{k}.
\]
A direct calculation gives
\begin{equation}\label{eq:general-exact}
\frac{\partial}{\partial z}\left(\frac{Q(z)}{(z-t)^{k-1}}F\right)
= -\frac{k-1}{k}
\left(
(k+1)Q(t)s^{-k}+\sum_{j=1}^{k}\frac{k+1-j}{j!}Q^{(j)}(t)s^{j-k}
\right)F .
\end{equation}
The term involving $Q^{(k+1)}$ cancels because $Q$ has degree $k+1$.  Integrating
\eqref{eq:general-exact} over a closed cycle and rewriting
$\int s^{-(k-j)}F\,dz$ as $I_\gamma^{(k-j)}(t)/(\alpha)_{k-j}$ gives exactly
\eqref{eq:general-pf}.
\end{proof}

For $k=2$, equation \eqref{eq:general-pf} becomes, after division by a constant,
\[
        Q(t)I''(t)+Q'(t)I'(t)+\frac{Q''(t)}8I(t)=0,
\]
which is the homogeneous part of the Shapiro--Tater equation for the classical
Heun spectral problem, see \cite{ShT}.  The normalized exterior Cauchy transform in the classical
case satisfies the inhomogeneous equation
\begin{equation}
\label{eq:classical-pf}
        Q(t)C''(t)+Q'(t)C'(t)+\frac{Q''(t)}8 C(t)+\frac{Q'''(t)}{24}=0,
\end{equation}
see \cite{ShT,STT}.  The inhomogeneous term records the normalization
$C(t)=t^{-1}+O(t^{-2})$ at infinity.

\medskip
We now spell out the first non-classical case $k=3$.  Let
\begin{equation}
\label{eq:quartic-Q}
        Q(t)=t^4+a_3t^3+a_2t^2+a_1t+a_0.
\end{equation}
For the cyclic triple cover
\[
        y^3=\frac{z-t}{Q(z)}
\]
and the differential
\[
        \eta_t=(z-t)^{-2/3}Q(z)^{-1/3}\,dz,
\]
Proposition~\ref{prop:general-pf} gives the following explicit equation.

\begin{corollary}[The cubic WKB Picard--Fuchs operator]
\label{cor:k3-pf}
For every locally constant cycle $\gamma$ on the above triple cover, the period
$I_\gamma(t)=\int_\gamma\eta_t$ satisfies
\begin{equation}
\label{eq:k3-pf}
        81Q(t)I_\gamma'''(t)+162Q'(t)I_\gamma''(t)
        +90Q''(t)I_\gamma'(t)+10Q'''(t)I_\gamma(t)=0 .
\end{equation}
\end{corollary}

If the normalized Cauchy transform of the selected spectral mother body is governed
by this WKB Picard--Fuchs operator, then its analogue of \eqref{eq:classical-pf}
would be
\begin{equation}
\label{eq:k3-inhom}
        81Q(t)C'''(t)+162Q'(t)C''(t)+90Q''(t)C'(t)+10Q'''(t)C(t)+30=0,
\end{equation}
for monic quartic $Q$.  Indeed, substituting
$C(t)=t^{-1}+O(t^{-2})$ into the left-hand side of the homogeneous operator in
\eqref{eq:k3-pf} gives the constant $-30$; all further terms are negative powers
of $t$.

\medskip
In the cubic case the WKB/Picard--Fuchs prediction is compatible with the
finite-band exterior germ in the following precise sense: for \(k=3\), the
finite-band averaging and the WKB Picard--Fuchs operator give the same exterior
differential equation.

\medskip
For reference we record the finite-band $k=3$ germ explicitly.  In the coordinate
where
\[
        Q_3(z)=z^4+a_3z^3+a_2z^2+a_1z,
\]
let
\[
        A=t+a_3\tau^3,
        \qquad B=(1-\tau^3)\tau^3a_2,
        \qquad D=(1-\tau^3)^2\tau^3a_1,
\]
and let $\Psi=\Psi(t,\tau)$ be the branch satisfying $\Psi=t+O(1)$ at infinity of
\begin{equation}
\label{eq:k3-cubic-psi}
        \Psi^3=A\Psi^2+B\Psi+D .
\end{equation}
Then Theorem~\ref{th:2} gives
\begin{equation}
\label{eq:k3-finiteband-C}
        C_Q(t)=\int_0^1
        \frac{\Psi(t,\tau)\,d\tau}{3\Psi(t,\tau)^2-2A\Psi(t,\tau)-B}.
\end{equation}

\begin{proposition}[The cubic finite-band germ satisfies the cubic Picard--Fuchs equation]
\label{prop:k3-finiteband-pf}
Assume that the distinguished zero of the monic quartic leading coefficient has
been placed at the origin,
\[
        Q(t)=t^4+a_3t^3+a_2t^2+a_1t .
\]
Let \(C_Q(t)\) be the finite-band exterior germ
\[
        C_Q(t)=\int_0^1
        \frac{\Psi(t,\tau)\,d\tau}
        {3\Psi(t,\tau)^2-2A\Psi(t,\tau)-B},
\]
where
\[
        A=t+a_3\tau^3,\qquad
        B=(1-\tau^3)\tau^3a_2,\qquad
        D=(1-\tau^3)^2\tau^3a_1,
\]
and where \(\Psi=\Psi(t,\tau)\) is the branch satisfying
\(\Psi=t+O(1)\) at infinity of
\[
        \Psi^3=A\Psi^2+B\Psi+D .
\]
Then, as a germ at infinity, and hence by analytic continuation on the exterior
domain where \(C_Q\) is defined,
\begin{equation}
\label{eq:k3-finiteband-pf}
        81Q(t)C_Q'''(t)+162Q'(t)C_Q''(t)
        +90Q''(t)C_Q'(t)+10Q'''(t)C_Q(t)+30=0 .
\end{equation}
\end{proposition}

\begin{proof}
Put
\[
        b_\tau(w)=(1-\tau^3)w^{-1}
        -\tau^3(a_3+a_2w+a_1w^2).
\]
By the finite-band moment formula, the expansion of \(C_Q\) at infinity is
\[
        C_Q(t)=\frac1t+\sum_{m\ge1}\frac{M_m}{t^{m+1}},
        \qquad
        M_m=\int_0^1 \CT_w\, b_\tau(w)^m\,d\tau .
\]
We shall prove that these moments satisfy exactly the recurrence imposed by
\eqref{eq:k3-finiteband-pf}.

Expanding \(b_\tau(w)^m\), the constant term is obtained as follows.  Let
\(\ell\) be the number of factors \(-\tau^3a_2w\), and let \(r\) be the number
of factors \(-\tau^3a_1w^2\).  Then the number of factors
\((1-\tau^3)w^{-1}\) must be \(\ell+2r\), while the number of factors
\(-\tau^3a_3\) is \(m-2\ell-3r\).  Thus
\[
        M_m=\sum_{\substack{\ell,r\ge0\\2\ell+3r\le m}}
        c_m(\ell,r)\,
        a_3^{m-2\ell-3r}a_2^\ell a_1^r ,
\]
where
\[
        c_m(\ell,r)=
        \frac{(-1)^{m-\ell-2r}m!\,
        \Gamma(m-\ell-2r+\frac13)}
        {3\,\Gamma(m+\frac43)\,
        (m-2\ell-3r)!\,\ell!\,r!}.
\]
Indeed, after the substitution \(s=\tau^3\),
\[
        \int_0^1
        (1-\tau^3)^{\ell+2r}
        \tau^{3(m-\ell-2r)}\,d\tau
        =
        \frac13
        B\!\left(m-\ell-2r+\frac13,\ell+2r+1\right),
\]
which gives the displayed coefficient.

Now substitute
\[
        C_Q(t)=\sum_{m\ge0}M_m t^{-m-1},\qquad M_0=1,
\]
into the differential operator
\[
        \mathcal L
        =
        81Q(t)\frac{d^3}{dt^3}
        +162Q'(t)\frac{d^2}{dt^2}
        +90Q''(t)\frac{d}{dt}
        +10Q'''(t).
\]
The constant term of \(\mathcal L C_Q\) is \(-30\).  For \(m\ge1\), the
coefficient of \(t^{-m}\) in \(\mathcal L C_Q\) is \(-3R_m\), where
\[
\begin{split}
        R_m={}&
        (3m-5)(3m-2)(3m+1)M_m \\
        &+a_3(3m-5)(3m-2)^2M_{m-1} \\
        &+3a_2(m-1)(3m-5)(3m-4)M_{m-2} \\
        &+27a_1(m-2)^2(m-1)M_{m-3},
\end{split}
\]
with the convention \(M_j=0\) for \(j<0\).  Therefore it remains only to prove
\(R_m=0\) for every \(m\ge1\).

Because the \(M_m\)'s are polynomials in \(a_1,a_2,a_3\), it is enough to check
the coefficient of each monomial
\[
        a_3^{m-2\ell-3r}a_2^\ell a_1^r .
\]
The four possible contributions are
\[
        c_m(\ell,r),\qquad
        c_{m-1}(\ell,r),\qquad
        c_{m-2}(\ell-1,r),\qquad
        c_{m-3}(\ell,r-1),
\]
where inadmissible terms are understood as zero.  For admissible indices one has
\[
        \frac{c_{m-1}(\ell,r)}{c_m(\ell,r)}
        =
        -\frac{(m+\frac13)(m-2\ell-3r)}
        {m(m-\ell-2r-\frac23)},
\]
\[
        \frac{c_{m-2}(\ell-1,r)}{c_m(\ell,r)}
        =
        -\frac{\ell(m+\frac13)(m-\frac23)}
        {m(m-1)(m-\ell-2r-\frac23)},
\]
and
\[
        \frac{c_{m-3}(\ell,r-1)}{c_m(\ell,r)}
        =
        -\frac{r(m+\frac13)(m-\frac23)(m-\frac53)}
        {m(m-1)(m-2)(m-\ell-2r-\frac23)} .
\]
Substituting these three ratios into the coefficient of
\(a_3^{m-2\ell-3r}a_2^\ell a_1^r\) in \(R_m\), and bringing to a common
denominator, gives the identically zero numerator
\[
\begin{split}
& (3m-5)(3m-2)(3m+1) \\
&\quad
-(3m-5)(3m-2)^2
        \frac{(m+\frac13)(m-2\ell-3r)}
        {m(m-\ell-2r-\frac23)} \\
&\quad
-3\ell(m-1)(3m-5)(3m-4)
        \frac{(m+\frac13)(m-\frac23)}
        {m(m-1)(m-\ell-2r-\frac23)} \\
&\quad
-27r(m-2)^2(m-1)
        \frac{(m+\frac13)(m-\frac23)(m-\frac53)}
        {m(m-1)(m-2)(m-\ell-2r-\frac23)}
        =0 .
\end{split}
\]
Hence \(R_m=0\) for all \(m\ge1\).  Consequently
\[
        \mathcal L C_Q(t)+30=0
\]
as a Laurent series at infinity.  Since \(C_Q\) is analytic as an exterior germ,
the identity holds on its domain of analytic continuation.  This proves
\eqref{eq:k3-finiteband-pf}.
\end{proof}

Finally, the support prediction in the case $k=3$ becomes especially concrete.  The
curve \eqref{eq:pf-curve} is a cyclic triple cover with branch points at $z=t$ and
at the four zeros of $Q$.  When $t$ approaches a zero of $Q$, one active cycle
collapses; hence the leaves of the predicted tree are the four zeros of $Q$.  Its
edges should be components of real-period loci
\begin{equation}
\label{eq:k3-period-locus}
        \Re \int_\gamma (z-t)^{-2/3}Q(z)^{-1/3}\,dz=0,
\end{equation}
for admissible cycles $\gamma$, with density proportional to the variation of the
corresponding imaginary period.  Thus, for $k=3$, the proposed tree should be
reconstructible from level curves of solutions of the third-order equation
\eqref{eq:k3-pf}, together with the positivity condition in \eqref{eq:wkb-jump}.

\section{Finite-band potentials and mother bodies}
\label{sec:mother-body}

We now clarify the role of the averaged measure $\nu_\dq$ introduced above.  The finite-band recurrence does not necessarily give the limiting root-counting measure of the spectral polynomials themselves.  What it gives, by Theorems~\ref{th:1} and~\ref{th:2}, is the limiting exterior logarithmic potential on $\Omega_Q$, or equivalently all holomorphic moments together with their analytic continuation to the exterior of $K_Q$.  This distinction is essential when the roots of $Q_k$ are not real.

Indeed, for a typical complex polynomial $Q_k$ the support of
\[
        \nu_\dq=\int_0^1\nu_{Q_k}(\tau)d\tau
\]
need not be one-dimensional.  Since the supports of the frozen measures $\nu_{Q_k}(\tau)$ vary with $\tau$, their union can fill a two-dimensional region.  This is incompatible with the numerical pictures for the spectral polynomials in the case $r=1$, where the roots of $Sp_n$ appear to accumulate on a finite tree.  Thus the correct interpretation should be the following: the averaged finite-band measure $\nu_\dq$ determines an exterior potential, while the limiting zero measure of $Sp_n$ is a one-dimensional mother body for this potential.

\medskip

We therefore propose the following refinement of Conjecture~\ref{conj:main}.

\begin{conjecture}[Mother-body form of the spectral limit]
\label{conj:mother-body}
Let $\dq$ be a high-order Heun operator with Fuchs index one and with generic
leading coefficient $Q_k$.  The root-counting measures of $Sp_n$ converge weakly
to a positive measure $\mu_{Q_k}$ depending only on $Q_k$.  Its support is a
finite planar tree
\[
        \Gamma_{Q_k}\subset K_Q,
\]
whose leaves are the roots of $Q_k$.  Moreover
\[
        U^{\mu_{Q_k}}(t)=U^{\nu_\dq}(t)
\]
in the exterior of $K_Q$; equivalently, $\mu_{Q_k}$ has the same exterior Cauchy
transform as the averaged finite-band measure $\nu_\dq$.  Thus $\mu_{Q_k}$ is a
positive mother body for the exterior field determined by the finite-band
averaging procedure.
\end{conjecture}

This formulation is in line with the potential-theoretic language of mother bodies and critical measures, see for example \cite{BoShMother,MFR,KS}.  It also reconciles two apparently conflicting pieces of evidence.  On the one hand, the Kuijlaars--Van Assche type averaging procedure naturally gives a family of frozen recurrence measures and hence an averaged object which may have two-dimensional support.  On the other hand, the actual zeros of $Sp_n$ are observed to lie close to a one-dimensional tree.  Equality of exterior potentials allows both statements to be true simultaneously.

The mother-body viewpoint also explains why the real-rooted case is exceptional.  If all roots of $Q_k$ are real, the frozen recurrence measures live on subsets of the real line and the averaged measure is already one-dimensional.  In that case there is no need to collapse a two-dimensional potential to a skeleton; this is consistent with the equality $\nu_\dq=\mu_Q$ proved in the real case in \cite{BoSh}.

\begin{problem}[Mother body from the finite-band determinant]
\label{prob:finite-band-mother}
Prove that the exterior potential determined by the finite-band recurrence admits a positive mother body supported on a tree with leaves at the zeros of $Q_k$.  Then prove that the zeros of $Sp_n$ select this mother body rather than the averaged measure $\nu_\dq$ itself.
\end{problem}

A possible approach to Problem~\ref{prob:finite-band-mother} is to study the limit
\[
        \lim_{n\to\infty}\frac{1}{n+1}\log |Sp_n(t)|
\]
directly as a solution of an obstacle-type problem.  The finite-band recurrence gives its expansion at infinity.  The missing part is a variational characterization selecting the smallest positive carrier of this exterior potential.  If such a characterization has the usual $S$-property, then the support should be a critical graph of a quadratic differential associated with the mother body, not with the Cauchy transform equation for a fixed Stieltjes polynomial.

The next observation isolates the remaining step in upgrading subsequential convergence to convergence.

\begin{proposition}[Reduction of convergence to uniqueness]
\label{prop:uniqueness-reduction}
Assume that every subsequential weak limit of the measures $\mu_n$ is supported on a compact set in $\Conv(Q_k)$ with connected complement and empty interior, and assume that there is at most one positive measure in this class whose logarithmic potential agrees near infinity with the finite-band potential $U_Q^{\rm ext}$.  Then the whole sequence $\mu_n$ converges weakly.  Its limit is this unique measure.

In particular, Conjecture~\ref{conj:mother-body} would imply the existence part of Conjecture~\ref{conj:main} once the corresponding positive tree mother body is shown to be unique.
\end{proposition}

\begin{proof}
The localization theorem quoted in the introduction makes the family $\{\mu_n\}$ tight.  Hence every subsequence has a further weakly convergent subsequence.  By Theorem~\ref{th:1}, any such subsequential limit has the same holomorphic moments, equivalently the same logarithmic potential in a neighbourhood of infinity.  The assumed uniqueness therefore forces all subsequential limits to coincide.  This is equivalent to weak convergence of the whole sequence.
\end{proof}

\section{A WKB route to the spectral tree}
\label{sec:wkb}

Although the previous section warns against identifying the spectral tree with the Stokes graph for one fixed Van Vleck polynomial, WKB ideas still seem to give the most promising route to the missing uniqueness statement.  The point is that the spectral polynomial should be viewed as recording the set of Van Vleck parameters for which the WKB problem admits a polynomially quantized solution.

Let
\[
        V_n(z)=-\lambda_n(z-t_n)
\]
be an admissible Van Vleck polynomial and let $S_n$ be the corresponding Stieltjes polynomial of degree $n$.  Put
\[
        \mathcal C_n(z)=\frac1n\frac{S_n'(z)}{S_n(z)}.
\]
Away from the zeros of $S_n$ and from the zeros of $Q_k$, the formal WKB rule
\[
        \frac{1}{n^i}\frac{S_n^{(i)}(z)}{S_n(z)}=\mathcal C_n(z)^i+O(n^{-1})
\]
turns the differential equation
\[
        (\dq-\lambda_n(z-t_n))S_n=0
\]
into the algebraic equation
\begin{equation}
\label{eq:wkb-cauchy}
        Q_k(z)C(z)^k=z-t,
\end{equation}
where $t=\lim t_n$ and $C$ is a limiting Cauchy transform of the zeros of $S_n$.  Thus the leading WKB differential attached to the spectral parameter $t$ is
\begin{equation}
\label{eq:wkb-differential}
        \omega_t=\left(\frac{z-t}{Q_k(z)}\right)^{1/k}dz,
\end{equation}
with the branch normalized by $\omega_t=(z^{-1}+O(z^{-2}))dz$ at infinity.  The lower coefficients of $\dq$ enter only in lower WKB orders.  This agrees with Theorem~\ref{th:1}: the leading exterior potential of the spectral roots depends only on $Q_k$.

Equation~\eqref{eq:wkb-cauchy} is the fixed-$t$ WKB equation.  The spectral problem asks for which $t$ the associated algebraic differential can be completed to a global polynomial solution.  The usual Bohr--Sommerfeld picture suggests the following necessary condition.  Let
\[
        \Pi_\gamma(t)=\int_\gamma \omega_t
\]
be a period of the algebraic differential on the curve
\[
        y^k=\frac{z-t}{Q_k(z)}.
\]
If $t$ is an accumulation point of Van Vleck roots, then the leading exponential factors of the WKB solutions must be single-valued in modulus along the relevant cuts.  Equivalently, for every active cycle one expects
\begin{equation}
\label{eq:period-real}
        \Re \Pi_\gamma(t)=0.
\end{equation}
At finite degree this condition should be replaced by a Bohr--Sommerfeld quantization condition of the form
\begin{equation}
\label{eq:bohr-sommerfeld}
        n\,\Pi_\gamma(t)+\Pi_\gamma^{(1)}(t)+O(n^{-1})
        \in 2\pi i\,\mathbb Z,
\end{equation}
where $\Pi_\gamma^{(1)}$ is a subleading correction depending on the lower coefficients of $\dq$ and on endpoint exponents.  Taking real parts and letting $n\to\infty$ gives \eqref{eq:period-real}.  Thus the lower coefficients should move individual spectral roots by $O(1/n)$ but should not move their limiting carrier.

This leads to a concrete prediction for the tree in the Van Vleck-parameter plane.

\begin{conjecture}[WKB description of the spectral tree]
\label{conj:wkb-tree}
For generic $Q_k$ with simple roots, the support $\Gamma_{Q_k}$ of the limiting spectral measure is the union of the admissible Stokes-period arcs
\begin{equation}
\label{eq:wkb-tree-locus}
        \Gamma_{Q_k}=\overline{\bigcup_\gamma E_\gamma},
\end{equation}
where $E_\gamma$ is the set of all $t\in\Conv(Q_k)$ such that
$\Re \Pi_\gamma(t)=0$ and the corresponding WKB cut system is positive.
Here the union is taken only over cycles which occur in an admissible critical graph of the differential \eqref{eq:wkb-differential}.  The endpoints of the tree occur when an active cycle collapses, and these endpoints are precisely the zeros of $Q_k$.
\end{conjecture}

The phrase ``positive'' in Conjecture~\ref{conj:wkb-tree} is important.  For a candidate cut system $\Gamma$ in the $z$-plane, the limiting zero density of the Stieltjes polynomials would be given by the jump formula
\begin{equation}
\label{eq:wkb-jump}
        d\sigma_t(z)=\frac{1}{2\pi i}
        \big(C_+(z)-C_-(z)\big)\,dz,
\end{equation}
where $C^k=(z-t)/Q_k(z)$.  Only those cuts for which \eqref{eq:wkb-jump} is a positive measure should contribute to the spectral locus.  In the canonical coordinate
\[
        W_t(z)=\int^z\omega_t,
\]
these cuts should be horizontal or vertical trajectories, depending on the convention, of the $k$-differential \eqref{eq:wkb-differential}.  The $S$-property of the resulting mother body is the potential-theoretic translation of the equality of the real parts of the WKB actions on the two sides of a cut.

The same formalism also predicts the density of the limiting measure on an edge of the spectral tree.  Suppose that along a smooth edge $E$ of $\Gamma_{Q_k}$ a single period $\Pi_\gamma(t)$ is active and that $\Re \Pi_\gamma(t)=0$ defines the edge locally.  The quantization condition \eqref{eq:bohr-sommerfeld} then suggests that the number of Van Vleck roots on a subarc of $E$ is asymptotic to the variation of the imaginary part of the active period.  In other words, the candidate line density is
\begin{equation}
\label{eq:wkb-density}
        d\mu_E(t)=\frac{1}{2\pi}
        \left|d\,\Im \Pi_\gamma(t)\right|,
\end{equation}
with the obvious modification when several periods are active.  The total density obtained from \eqref{eq:wkb-density} must then be normalized and checked against the moment formula \eqref{eq:moment-formula}.  This gives a practical test for the WKB prediction: the tree obtained from periods must reproduce the exterior Cauchy transform \eqref{eq:Cauchy-nu}.

The resulting strategy for proving existence of the limit of spectral polynomials is therefore the following.
\begin{enumerate}
\item Prove WKB asymptotics for polynomial solutions uniformly for $t$ away from Stokes transitions, with error estimates strong enough to justify \eqref{eq:bohr-sommerfeld}.
\item Classify the admissible positive Stokes graphs as $t$ varies in $\Conv(Q_k)$.
\item Show that the quantized solutions become equidistributed on the locus \eqref{eq:wkb-tree-locus} with density \eqref{eq:wkb-density}.
\item Verify that the resulting measure has the exterior moments \eqref{eq:moment-formula}.  By Proposition~\ref{prop:uniqueness-reduction}, uniqueness of the positive mother body then proves convergence of the whole sequence $\mu_n$.
\end{enumerate}

At present this should be regarded as a program rather than a proof.  Its advantage is that it predicts not only the existence of the limiting spectral measure, but also the geometry of its support.  The edges should be period-level curves in the Van Vleck parameter $t$, vertices should occur at Stokes transitions where two or more admissible periods become active, and the leaves should be the roots of $Q_k$.  This is exactly the tree-like behavior seen in Figures~\ref{fig2} and ~\ref{fig1}.

\section{Final Remarks}

\subsection*{1.} The discussion above changes the interpretation of the difficulty in the complex cubic case.  The problem is not a failure of Theorem~\ref{th:mainrec}; rather it is the expected non-uniqueness of a compactly supported measure with prescribed holomorphic moments in the plane.  That theorem gives only the exterior potential: the finite-band averaging procedure naturally produces a two-dimensional auxiliary measure, whereas the actual roots of the spectral polynomials appear to select a one-dimensional mother body for the same exterior potential; see Figure~\ref{fig2}.

\begin{problem}
For a complex cubic polynomial $Q(z)$, prove that the roots of the spectral polynomials converge to a positive mother-body measure for the exterior potential produced by the finite-band recurrence.  Describe its support and compare it with the compact obtained in \cite{ShT,STT}.
\end{problem}

Based on the numerical evidence and on the preceding section, we expect the following more precise picture:
\begin{itemize}
\item[(i)] the exterior potential of the limiting distribution of Van Vleck roots is independent of the lower coefficients of the operator and depends only on the leading polynomial $Q(z)$;
\item[(ii)] the actual limiting root-counting measure is a positive mother body for this exterior potential;
\item[(iii)] its support is a finite tree inside $\Conv(Q)$ with leaves at the roots of $Q$.
\end{itemize}

The moment theorem gives a concrete test for any proposed limiting tree.  If a candidate tree $\Gamma$ with a positive density $\rho(s)|ds|$ is proposed, then it must solve the infinite moment system
\[
        \int_\Gamma z^m\rho(s)|ds|=M_m(Q),\qquad m=0,1,2,\ldots,
\]
where $M_0=1$ and $M_m(Q)$ is given by \eqref{eq:moment-formula}.  This is a useful computational route: first reconstruct the exterior Cauchy transform from \eqref{eq:Cauchy-nu}, then search for the smallest positive carrier with these moments.  The expected carrier is the mother-body tree.

\begin{problem}[Moment reconstruction of the spectral tree]
Use the moment formula \eqref{eq:moment-formula} to characterize the unique positive tree measure, if it exists, whose support has leaves at the zeros of $Q_k$ and whose exterior potential is given by \eqref{eq:Cauchy-nu}.  In particular, decide whether positivity and the $S$-property force uniqueness of the tree.
\end{problem}

\begin{problem}[WKB periods versus finite-band moments]
For a fixed generic polynomial $Q_k$, compute the period locus \eqref{eq:wkb-tree-locus} and the WKB density \eqref{eq:wkb-density}.  Prove, at least in the cubic case, that its Cauchy transform has the same expansion at infinity as \eqref{eq:Cauchy-nu}.  This would identify the WKB tree with the mother body selected by the finite-band determinant.
\end{problem}

\begin{problem}[Discontinuity formula for the mother body]
Assume that the limiting spectral measure is supported on a finite union of analytic arcs.
Derive its density as the jump of the exterior Cauchy transform \eqref{eq:Cauchy-nu}
across these arcs.  In the cubic case this should turn the integral formula in the first
example into a scalar Riemann--Hilbert problem.  A solution of this problem would give a
practical route from the finite-band potential to the observed spectral tree.
\end{problem}

\begin{problem}[Universality with respect to lower coefficients]
Theorems~\ref{th:1} and~\ref{th:2} prove independence of the lower coefficients of $\dq$ at the level of the whole exterior potential.  Prove, or disprove by numerical examples, the stronger
statement that the lower coefficients do not affect the selected mother body.  Equivalently,
show that the passage from the exterior potential to the positive tree measure is canonical
for the class of spectral polynomials arising from high-order Heun operators.
\end{problem}

The preceding problems seem to be the most natural next steps.  The analytic part
should start with the explicitly computable cubic case.  The structural part asks
whether the finite-band moment theorem already contains the whole limiting object,
or only its exterior shadow.

\subsection*{2.} Even more difficulties occur when dealing with the more general 
Lam\'e equation \eqref{eq:comLame} since in this case van Vleck polynomials 
and their roots cannot be found by means of a determinantal equation. They
are in fact related to a more complicated situation when the rank of a 
certain non-square matrix is less than the maximal one, see \cite{Sh}. 

\begin{problem}
Describe the asymptotic distribution of the roots of all Van Vleck 
polynomials for equation \eqref{eq:comLame} when $n\to\infty$.
\end{problem} 

Figure~\ref{fig3} illustrates this asymptotic distribution. 

\begin{figure}[H]
\begin{center}
\IfFileExists{picQ4d2V.pdf}{\includegraphics[width=0.62\linewidth]{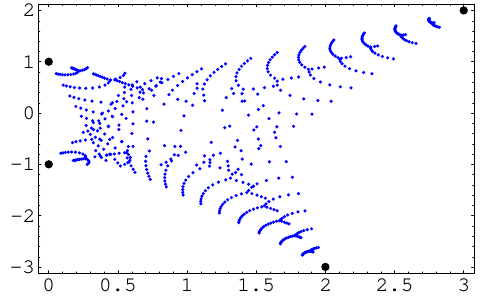}}{\fbox{\parbox{0.75\linewidth}{Figure file \texttt{picQ4d2V.pdf} not found.}}}
\end{center}
\caption{The union of the roots of 861 quadratic van Vleck polynomials 
corresponding to Stieltjes polynomials of degree $40$ for the classical 
Lam\'e equation $Q(z)S''(z)+\frac{Q'(z)}{2}S'(z)+V(z)S(z)=0$ with 
$Q(z)=(z^2+1)(z-3i-2)(z+2i-3)$.}
\label{fig3}
\end{figure}

\begin{ack}
The author  thanks his coauthor and friend of many years, Milo\v{s} Tater,  for his interest in this project.
\end{ack}

\end{document}